\documentclass[envcountset]{llncs}

\usepackage{amsmath, amssymb, stmaryrd}
\usepackage{comment,booktabs,enumerate}
\usepackage{paralist}

\usepackage{graphicx}
\usepackage{mathtools}
\usepackage{tikz}
\usetikzlibrary{arrows}
\usepackage{caption}
\usepackage{float}

\usepackage{algorithm}
\usepackage[noend]{algpseudocode}
\usepackage{amsmath, amssymb, stmaryrd}
\usepackage{xspace}
\usepackage{xcolor}
\usepackage{needspace}
\usepackage{enumerate}
\usepackage{ifthen}
\usepackage{fancyhdr}
\usepackage{algorithmicx}
  \usetikzlibrary{arrows,automata}
\usepackage{listings}
  
  \lstdefinelanguage{pseudo}{
    morekeywords={if,elseif,then,return,end,choose,guess,when,for,foreach,case},
    morekeywords=[3]{false,true,and,or,not},
    morecomment=[l]{//}
  }
\newcommand{\logicFont}[1]{\protect\ensuremath{\mathrm{#1}}
}
\newcommand{\problemFont}[1]{\protect\ensuremath{\mathsf{#1}}
}
\newcommand{\cmdFont}[1]{\protect\ensuremath{\mathrm{#1}}\xspace}
\newcommand{\classFont}[1]{\protect\ensuremath{\mathsf{#1}}\xspace}
\newcommand{\clneg}{\mathord{\sim}}

\newcommand{\PD}{\logicFont{PD}}
\newcommand{\PL}{\logicFont{PL}}

\newcommand{\PLT}{\logicFont{PL}[\clneg]}

\newcommand{\PLIT}{\logicFont{PL}[\bot_{\rm c},\clneg]}
\newcommand{\PLI}{\logicFont{PL}[\bot_{\rm c}]}
\newcommand{\PLIncT}{\logicFont{PL}[\sub,\clneg]}
\newcommand{\PLInc}{\logicFont{PL}[\sub]}

\newcommand{\PLIncIT}{\logicFont{PL}[\bot_{\rm c},\sub,\clneg]}

\newcommand{\SAT}{\problemFont{SAT}}

\newcommand{\TQBF}{\problemFont{TQBF}}
\newcommand{\MC}{\problemFont{MC}}

\newcommand{\VAL}{\problemFont{VAL}}

\newcommand{\NEXPTIME}{\classFont{NEXPTIME}}
\newcommand{\EXPTIME}{\classFont{EXPTIME}}
\newcommand{\EXP}{\classFont{EXP}}

\newcommand{\PSPACE}{\classFont{PSPACE}}

\newcommand{\NP}{\classFont{NP}}
\newcommand{\NC}{\classFont{NC}}
\newcommand{\PTIME}{\classFont{P}}

\newcommand{\AEP}{\classFont{AEXPTIME(poly)}}
\newcommand{\AEPoly}{\classFont{APTIME}}
\newcommand{\coNP}{\classFont{coNP}}
\newcommand{\coNEXP}{\classFont{co}\classFont{NEXPTIME}}
\newcommand{\siglevel}[1]{\Sigma^{\EXP}_{#1}}
\newcommand{\pilevel}[1]{\Pi^{\EXP}_{#1}}

\newcommand{\dep}[1]{\cmdFont{dep}\!\left(#1\right)}
\newcommand{\fdep}[1]{\cmdFont{=}\!\left(#1\right)}

\newcommand{\ddfn}{\mathrel{\mathop{{\mathop:}{\mathop:}}}=}

\newcommand{\set}[3][]{\protect\ensuremath{\left\{#2\;\middle|\;\ifthenelse{\equal{#1}{}}{\text{#3}}{\parbox{#1}{#3}}\right\}}}

\newcommand{\var}[1]{\textrm{Var}(#1)}
\newcommand{\al}{\alpha}

\newcommand{\tuple}[1]{\vec{#1}}
\newcommand {\indep}[3] {#2 ~\bot_{#1}~ #3}
\newcommand {\indepc}[2] {#1 ~\bot~ #2}

\newcommand{\sub}{\subseteq}
\newcommand{\ma}[1]{\cmdFont{max}\!\left(#1\right)}

\newcommand{\intimp}{\multimap}
\newcommand{\cvee}{\hspace{.5mm}\varovee\hspace{.5mm}}
\newcommand{\cwedge}{\otimes}

\newcommand{\cor}[1]{\mathrm{bin}(#1)}

\begin{document}

\markboth{M. Hannula et al.}{Complexity of Propositional  Logics in Team Semantics}

\title{Complexity of Propositional Logics in Team Semantics}

\author{Miika Hannula\inst{1} \and Juha Kontinen\inst{1}  \and Jonni Virtema\inst{2} \and Heribert Vollmer \inst{2} }

\institute{University of Helsinki, Department of Mathematics and Statistics, Helsinki, Finland \email{\{miika.hannula,juha.kontinen\}@helsinki.fi} 
\and 
{Leibniz Universit\"at Hannover, Institut f\"ur Theoretische Informatik, Germany \texttt{jonni.virtema@gmail.com}, \texttt{vollmer@thi.uni-hannover.de}}}
\maketitle
\begin{abstract}
We  classify the computational  complexity of the satisfiability, validity and model-checking problems  for propositional independence, inclusion, and team logic.  Our main result shows that the satisfiability and validity problems for propositional team logic are complete for alternating exponential-time with polynomially many alternations.
\end{abstract}


\section{Introduction}

Dependence logic \cite{vaananen07} is a new logical framework for formalising and studying  various notions of dependence and independence that are important in many scientific disciplines such as experimental physics, social choice theory, computer science, and cryptography.  Dependence logic extends first-order logic by dependence atoms
\begin{equation}
\dep{x_1,\dots,x_n,y}
\end{equation}
expressing that the value of the variable $y$ is functionally determined on the values of $x_1,\dots,x_n$. 
Satisfaction for formulas of dependence logic is defined using sets of assignments (\emph{teams}) and not in terms of  single assignments as in first-order logic. 
Whereas dependence logic studies the notion of functional dependence, independence and inclusion logic (introduced in  \cite{DBLP:journals/sLogica/GradelV13} and
\cite{galliani12}, respectively) formalize the concepts of independence and inclusion.
Independence logic (inclusion logic) is obtained from dependence logic by
replacing dependence atoms by the so-called independence atoms $ \indep{\tuple y}{\tuple x}{\tuple z}$ (inclusion atoms $\tuple x \sub \tuple y$).
The intuitive meaning of the independence atom is that the variables of the tuples $\tuple x$ and $\tuple z$ are independent of each other for any fixed value of the  variables in $\tuple y$, whereas the inclusion atom declares that all values of the tuple $\tuple x$ appear also as values of $\tuple y$. In database theory these atoms correspond to the so-called embedded multivalued dependencies and inclusion dependencies (see, e.g.,   \cite{DBLP:conf/foiks/HannulaK14}). Independence atoms have also a close connection to conditional independence in statistics. 

The topic of this article is propositional team semantics which has received relatively little attention so far. On the other hand, modal team semantics has been studied actively. Since the propositional logics studied in the article are fragments of  the corresponding modal logics, some upper bounds trivally transfer to the propositional setting.
The study of propositional team semantics as a subject of independent interest was initiated after  surprising connections between  propositional team semantics and the so-called \emph{inquisitive  semantics}  was discovered (see \cite{Yangthesis} for  details). The first systematic study on the expressive power of propositional dependence logic and many of its variants is due to  \cite{Yangthesis,DBLP:journals/corr/YangV14}. In the same works  natural deduction type inference systems for these logics are also developed, whereas in  \cite{DBLP:journals/corr/SanoV14}  a complete Hilbert-style axiomatization and a labeled tableaux calculus for propositional dependence logic is presented. Very recently Hilbert-style proof systems for related logics that incorporate the classical negation have been introduced by L\"uck, see \cite{Luck16a}.

  The computational aspects of (first-order) dependence logic and its variants have been actively studied, and are now quite well understood (see \cite{Durand2016}). On the other hand, the complexity of the propositional versions of these logics  have not been systematically studied. The study was initiated in \cite{DBLP:journals/corr/Virtema14} where the validity problem of propositional dependence logic was shown to be $\NEXPTIME$-complete. Also recently propositional inclusion logic has been studied in the article \cite{hellakmv15} and in the manuscript \cite{hkmv16}. In this article we study the complexity of satisfiability, validity and model-checking of propositional independence, inclusion and team  logic that extends propositional logic by  the classical negation. The classical negation has turned out to be  a very powerful connective in the settings of  first-order and modal team semantics, see e.g., \cite{DBLP:journals/corr/KontinenMSV14a} and \cite{DBLP:journals/fuin/KontinenN11}. Our results (see Table \ref{newresults}) show that the same is true in the propositional setting. In particular, our main result shows that the satisfiability and validity problems of team logic are complete for alternating exponential time with polynomially many alternations ($\AEP$). The results hold also for the extensions of propositional inclusion and independence logic by the classical negation.  Recently levels of the exponential hierarchy have been logically characterized in the context of propositional team semantics, in \cite{Luck16b,hakoluvi16}.

\begin{table}[!t]
\caption{Overview of the results (completeness results if not stated otherwise)}
\begin{center}
{
\begin{tabular}{cccc}\toprule
	 & $\SAT$ & $\VAL$ & $\MC$  \\\midrule
$\PLI$ 	&  $\NP$   & in $\coNEXP^{\NP}$ & $\NP$ \\
$\PLInc$ & $\EXPTIME$ \cite{hellakmv15} & $\coNP$ & $\PTIME$ \cite{hkmv16} \\
$\PLT,\PLIncIT$ & $\AEP$ & $\AEP$ & $\PSPACE$ \cite{muller14}\\\bottomrule
\end{tabular}
}
\end{center}
\label{newresults}
\end{table}

\section{Preliminaries}

In this section we define the basic concepts and results relevant to team-based propositional logics. We assume that the reader is familiar with propositional logic. 

\subsection{Syntax and semantics}
Let $D$ be a finite, possibly empty, set of proposition symbols. A function $s:D\to \{0,1\}$ is called an \emph{assignment}. A set $X$ of assignments $s:D\to \{0,1\}$ is called a \emph{team}. The set $D$ is the \emph{domain} of $X$. 
We denote by $2^D$ the set of \emph{all assignments} $s:D\to \{0,1\}$.

Let $\Phi$ be a set of proposition symbols. The syntax for propositional logic $\PL(\Phi)$ is defined as follows.
\[
\varphi \ddfn p\mid \neg p \mid (\varphi \wedge \varphi) \mid (\varphi \vee \varphi), \quad\text{ where $p\in\Phi$}.
\]
We write $\var{\varphi}$ for the set  of all proposition symbols that appear in $\varphi$. We denote by $\models_{\PL}$ the ordinary satisfaction relation of propositional logic defined via assignments in the  standard way. Next we give team semantics for propositional logic. 

\begin{definition}
Let $\Phi$ be a set of  proposition symbols and let $X$ be a team. The satisfaction relation $X\models \varphi$ is defined as follows. 
\begin{align*}
X\models p  \quad\Leftrightarrow\quad& \forall s\in X: s(p)=1.\\
X\models \neg p \quad\Leftrightarrow\quad& \forall s\in X: s(p)=0.\\
X\models (\varphi\land\psi) \quad\Leftrightarrow\quad& X\models\varphi \text{ and } X\models\psi.\\
X\models (\varphi\lor\psi) \quad\Leftrightarrow\quad& Y\models\varphi \text{ and } 
Z\models\psi,
\text{ for some $Y,Z$ such that $Y\cup Z= X$}.
\end{align*}
\end{definition}
Note that in team semantics $\neg$ is not the classical negation (denoted by $\clneg$ in this article) but a so-called \emph{dual} negation that does not satisfy the law of excluded middle. 
Next proposition shows that the team semantics and the ordinary semantics for propositional logic defined via assignments coincide.
\begin{proposition}[\cite{vaananen07}]
\label{PLflat}
Let $\varphi$ be a formula of propositional logic and let $X$ be a propositional team. Then 
\(
X\models \varphi \;\text{ iff }\; \forall s\in X: s\models_\mathit{PL} \varphi.
\)
\end{proposition}
The syntax of \emph{propositional dependence logic} $\PD(\Phi)$ is obtained by extending the syntax of $\PL(\Phi)$ by the  rule
\[
\varphi \ddfn \dep{p_1,\dots,p_n,q}, \quad\text{ where $p_1,\dots,p_n,q\in\Phi$}.
\]
The semantics for the propositional dependence atoms are defined as follows:
\begin{align*}
X\models \dep{p_1,\dots,p_n,q} \quad\Leftrightarrow\quad& \forall s,t\in X: s(p_1)=t(p_1), \dots, s(p_n)=t(p_n)\\
&\text{implies that } s(q)=t(q).
\end{align*}
The next proposition is very useful when determining the complexity of \PD, and it is proved analogously as for   first-order dependence logic \cite{vaananen07}.
%
\begin{proposition}[Downwards closure]\label{dcprop}
Let $\varphi$ be a PD-formula and let $Y\subseteq X$ be propositional teams. Then 
$X\models \varphi$ implies $Y\models \varphi$.
\end{proposition}

In this article we study the variants of $\PD$ obtained by replacing dependence atoms in terms of the so-called independence or inclusion atoms: The syntax of \emph{propositional independence logic} $\PLI(\Phi)$ is obtained by extending the syntax of $\PL(\Phi)$ by the rule
\[
\varphi \ddfn  \indep{\tuple p}{\tuple q}{\tuple r},
\]
 where $\tuple p$, $\tuple q$, and  $\tuple r$ are finite tuples of proposition symbols (not  necessarily of the same length). The  syntax of \emph{propositional inclusion logic} $\PLInc(\Phi)$ is obtained by extending the syntax of $\PL(\Phi)$ by the  rule
\[
\varphi \ddfn   \tuple p \sub \tuple q,
\]
where $\tuple p$ and $\tuple q$ are finite tuples of proposition symbols with the same length. 
 Satisfaction for these atoms is defined as follows. 
If $\tuple p =(p_1, \ldots ,p_n)$ and $s$ is an assignment, we write $s(\tuple p)$ for $\left(s(p_1),\dots,s(p_n)\right)$.
\begin{align*}
X\models \indep{\tuple p}{\tuple q}{\tuple r} \quad\Leftrightarrow\quad& \forall s,t\in X: \text{ if }s(\tuple p)=t(\tuple p)\\
&\text{then there exists } u\in X: u(\tuple p\tuple q)=s(\tuple p\tuple q) \text{ and }   u(\tuple r)=t(\tuple r).\\
X\models {\tuple p}\sub {\tuple q} \quad\Leftrightarrow\quad& \forall s\in X \exists t\in X :s(\tuple p)=t(\tuple q).
\end{align*}

It is easy to check that neither  $\PLI$ nor $\PLInc$  is a downward closed logic (cf. Proposition \ref{dcprop}). However, analogously to first-order inclusion logic \cite{galliani12}, the formulas of  $\PLInc$   have the following closure property. 

\begin{proposition}[Closure under unions]\label{closureunions}
Let $\varphi \in \PLInc$ and let $ X_i$, for $i\in I$, be teams. Suppose that $X_i\models \varphi$, for each $i\in I$. Then $\bigcup_{i\in I} X_i \models \varphi$. 
\end{proposition}

We will also consider the extensions of $\PL$, $\PLI$ and $\PLInc$, by the classical negation $\clneg$ with the standard semantics:
\[X\models \clneg \varphi \Leftrightarrow  X\not \models  \varphi. \]
These extensions are denoted by $\PL[\clneg]$ (propositional team logic),  $\PLIT$ and $\PLIncT$, respectively. 

A general notion of a \emph{generalized dependency atom} expressing a property of a propositional team has also been studied in the literature. For the purposes of this article  precise definitions are not required and are thus omitted, for a detailed exposition for generalised dependency atoms see, e.g., \cite{hakoluvi16}. We say that a generalized dependency atom $A$ has a polynomial time checkable semantics if $X\models A(\tuple p)$ can be decided in polynomial time with respect to the combined size of $X$ and $\tuple p$. Each of the atoms defined above are examples of generalized dependency atoms. It is easy to see that each of these atoms has a polynomial time checkable semantics.

\subsection{Auxiliary operators}\label{ao}
The following additional operators will be used in this paper:
\begin{align*}
X\models \varphi \cvee \psi \quad \Leftrightarrow \quad & X\models \varphi \textrm{ or } X \models \psi,\\
X \models \varphi \cwedge  \psi \quad \Leftrightarrow \quad & \forall\, Y,Z\sub X: \textrm{ if }Y\cup Z=X\textrm{, then } Y\models \varphi \textrm{ or } Z\models \psi,\\
X \models \varphi \intimp \psi\quad \Leftrightarrow \quad & \forall\, Y\sub X: \textrm{if } Y \models \varphi,\textrm{ then }  Y \models \psi,\\
X \models \ma{ x_1,\ldots ,x_n} \quad \Leftrightarrow \quad & \{ (s(x_1),\ldots ,s(x_n) )\mid s\in X\} =\{0,1\}^n.
\end{align*}
If $X \models \ma{\tuple x}$, we say that $X$ is \emph{maximal over $\tuple x$}. If tuples $\tuple x$ and $\tuple y$ are pairwise disjoint and $X \models \ma{\tuple x}\wedge  \indepc{\tuple x }{\tuple y}$, then we say that $X$ is \emph{maximal over $\tuple x$ for all $\tuple y$}. 
\begin{proposition}\label{trans}
The operators $\dep{\cdot}, \cvee, \cwedge, \intimp,$ and $\max(\cdot)$ have uniform polynomial size translations in  $\PL[\clneg]$.
\end{proposition}
\begin{proof}
We present the following translations of which item 3 is due to \cite{muller14} and item 4 uses the idea of \cite{abramsky09}. 
\begin{enumerate}
\item The connective $\cwedge$ is actually the dual of $\vee$, and hence $\varphi \cwedge \psi$ can be written as $ \clneg (\clneg \varphi \vee \clneg \psi)$.
\item Intuitionistic disjunction $\varphi \cvee \psi$ can be written as $\clneg (\clneg \varphi \wedge \clneg \psi)$.
\item Intuitionistic implication $\varphi \intimp \psi$ can be expressed as $  (\clneg \varphi \cvee \psi) \cwedge \clneg(p\vee\neg p)$.


\item First note that $\dep{x}$ can be written as $x \cvee \neg x$. Using this we can write $ \dep{x_1, \ldots ,x_n,y}$ as $\bigwedge_{i=1}^n \dep{x_i} \intimp \dep{y}$. 

\item We show that $\ma{ x_1,\ldots ,x_n}$ is equivalent to  
$\clneg \bigvee_{i=1}^n \dep{x_i}.$
Assume first that $X \models \bigvee_{i=1}^n \dep{x_i}$, we show that $X \not\models \ma{ x_1,\ldots ,x_n}$. By the assumption, we find $Y_1, \ldots ,Y_n \in X $, $\bigcup_{i=1}^n Y_i  = X$, such that $Y_i \models \fdep{x_i}$. Now for all $i$ there exists a $b_i\in \{0,1\}$ such that if $Y_i \neq \emptyset$, then for all $s\in Y_i$, $s(x_i) \neq b_i$. Since the assignment $x_i \mapsto b_i$ is not in $ X$, we obtain that $X \not\models  \ma{ x_1,\ldots ,x_n}$. 

Assume then that  $X \not\models  \ma{ x_1,\ldots ,x_n}$, we show that  $X \models \bigvee_{i=1}^n \dep{x_i}$. By the assumption there exists a boolean sequence $(b_1, \ldots ,b_n) $ such that for no $s \in X$ we have  $s(x_i) = b_i$ for all $i=1, \ldots ,n$. Let $Y_i:= \{s \in X \mid s(x_i) \neq b_i\}$. Since then $X = \bigcup_{i=1}^n Y_i$ and $Y_i \models \fdep{x_i}$, we obtain that $X \models \bigvee_{i=1}^n \dep{x_i}$.
\end{enumerate}
\qed
\end{proof}

\subsection{Satisfiability, validity, and model checking in team semantics}
Next we define satisfiability and validity in the context of team semantics. Let $L$ be a logic with  team semantics. A formula $\varphi\in L$  is \emph{satisfiable}, if there exists a non-empty  team $X$ such that $X\models\varphi$.
A formula $\varphi\in  L$  is \emph{valid}, if $X\models\varphi$ holds for every non-empty team $X$ such that the proposition symbols that occur in $\varphi$ are in the domain of $X$.\footnote{It is easy to show that all of the logics considered in this article have the so-called locality property, i.e., satisfaction of a formula depends only on the values of the proposition symbols that occur in the formula \cite{galliani12}.} Note that when the team is empty, satisfaction becomes easy to decide, see Proposition \ref{emptyprop} below. 

The satisfiability problem $\SAT(L)$ and the validity problem $\VAL(L)$ are then defined in the obvious manner: Given  a formula $\varphi \in L$, decide whether the formula is satisfiable (valid, respectively). The variant of the model checking problem that we are concerned with in this article is the following:
Given  a formula $\varphi \in L$  and a   team $X$, decide whether $X\models\varphi$. See Table \ref{table:PLoldresults} for known complexity results on $\PL$ and $\PD$.

\begin{proposition}\label{emptyprop}
Checking whether $\emptyset \models \varphi$, for $\varphi \in \PL[\bot_{\rm c}\sub, \clneg]$, can be done in $\PTIME$. Furthermore, $\emptyset \models \varphi$ for all $\varphi \in \PL[\bot_{\rm c}\sub]$.
\end{proposition}
\begin{proof}
Define a function $\pi:\PL[\bot_{\rm c},\sub, \clneg]\to\{0,1\}$ recursively as follows. Note that addition is $\hspace{-2mm}\mod 2$.
\begin{itemize}
\item If $\varphi\in\{p,\neg p, \indep{\tuple p}{\tuple q}{\tuple r},\tuple p \sub \tuple q\}$, then $\pi(\varphi)=1$.
\item If $\varphi = \psi_0 \wedge \psi_1$, then $\pi(\varphi)=\pi(\psi_0)\cdot \pi(\psi_1)$.
\item If $\varphi = \psi_0 \vee \psi_1$, then $\pi(\varphi)=\pi(\psi_0)\cdot \pi(\psi_1)$.
\item If $\varphi = \hspace{1mm}\clneg \psi$, then $\pi(\varphi) = \pi(\psi)+1$.
\end{itemize}
It is easy to check that $\emptyset \models \varphi$ iff $\pi(\varphi) =1$. Since $\pi(\varphi)$ can be computed in $\PTIME$, the claim follows.
\end{proof}


%

\begin{table}[!t]
\caption{Complexity of satisfiability, validity, and model checking of $\PL$ and $\PD$. All results are completeness results.}{
\begin{center}
\begin{tabular}{ccccc}\toprule
&       SAT  & VAL             & MC        & References\\ \midrule
\PL &   $\NP$     & co\NP    & $\NC^1$        & \cite{co71,Lev1973,Buss1987}      \\ 
\PD &   $\NP$     & \NEXPTIME        & $\NP$   &  \cite{lohvo13,ebbing12,DBLP:journals/corr/Virtema14}     \\    
\bottomrule
\end{tabular}
\end{center}
}
\label{table:PLoldresults}
\end{table}

\section{Complexity of Satisfiability and Validity}
In this section we consider the complexity of the satisfiability and validity problems for propositional independence, inclusion and team logic. 
\subsection{The logics  $\PLI$ and $\PLInc$}
We consider first the complexity of $\SAT(\PLI)$. The following simple lemma turns out to be very useful.

\begin{lemma}\label{PLIsingletons}
Let $\varphi \in \PLI$ and $X$ a team such that $X\models \varphi$. Then $\{s\}\models \varphi$, for all $s\in X$.
\end{lemma}
\begin{proof} The claim is proved using induction on the construction of  $\varphi$.  It is easy to check that a singleton team satisfies all independence atoms, and the cases  corresponding to disjunction and conjunction are straightforward. 
\end{proof}

\begin{theorem}
$\SAT(\PLI)$ is complete for $\NP$.
\end{theorem}
\begin{proof}
Note first that since $\SAT(\PL)$ is $\NP$-complete, it follows by Proposition \ref{PLflat} that $\SAT(\PLI)$ is  $\NP$-hard. For containment in $\NP$, note that  by Lemma \ref{PLIsingletons}, a formula
 $\varphi \in \PLI$ is satisfiable iff it is satisfied by some singleton team $\{s\}$. It is immediate that for any $s$, $\{s\}\models \varphi$ iff  $\{s\}\models \varphi^T$, where $\varphi^T\in \PL$ is acquired from $\varphi$
by replacing all independence atoms by $(p \vee \neg p)$. Thus it follows that $\varphi$ is satisfiable iff $\varphi^T$ is satisfiable. Therefore, the claim follows. \qed
\end{proof}
Next we consider the complexity of $\VAL(\PLI)$.
\begin{theorem}\label{valPLIhard}
$\VAL(\PLI)$ is hard for $\NEXPTIME$ and is in $\coNEXP^{\NP}$.
\end{theorem}
\begin{proof}
Since the dependence atom $\dep{\tuple x,y}$ is equivalent to the independence atom $\indep{\tuple x}{ y}{ y}$ and $\VAL(\PD)$ is  $\NEXPTIME$-complete \cite{DBLP:journals/corr/Virtema14}, hardness for $\NEXPTIME$ follows. We will show in Theorem \ref{mcind} on p.~\pageref{mcind} that the model checking problem for $\PLI$ is complete for $\NP$. It then follows that the complement of the problem $\VAL(\PLI)$ is in $\NEXPTIME^{\NP}$: 
the question whether $\varphi$ is in the complement of $\VAL(\PLI)$
 can be decided by guessing a subset $X$ of
$2^D$, where $D$ contains the set of proposition symbols appearing in $\varphi$, and checking whether $X\not\models \varphi$. Therefore $\VAL(\PLI)\in \coNEXP^{\NP}$.
\qed\end{proof}
Next we turn to propositional inclusion logic.
\begin{theorem}[\cite{hellakmv15}]
$\SAT(\PLInc)$ is complete for $\EXPTIME$.
\end{theorem}

We end this section by determining the complexity of  $\VAL(\PLInc)$.
\begin{theorem}
$\VAL(\PLInc)$ is complete for $\coNP$.
\end{theorem}
\begin{proof} Recall that $\PL$ is a sub-logic of $\PLInc$, and hence   $\VAL(\PLInc)$ is hard for $\coNP$. Therefore, it suffices to show    $\VAL(\PLInc)\in \coNP$. It is easy to check that, by Proposition \ref{closureunions}, a formula $\varphi\in \PLInc$ is valid iff it is satisfied by all singleton teams $\{ s\}$. Note also that, over  a singleton team $\{ s\}$, an inclusion atom
$ (p_1,\ldots ,p_n) \sub (q_1,\ldots ,q_n)$ is equivalent to the $\PL$-formula 
\[  \bigwedge_{1\le i\le n} p_i\leftrightarrow q_i. \]
Denote by $\varphi^*$ the $\PL$-formula acquired by replacing all inclusion atoms in $\varphi$ by their $\PL$-translations. By the above, $\varphi$ is valid iff $\varphi^*$ is valid. Since  $\VAL(\PL)$ is in $\coNP$ the claim follows. 
\qed
\end{proof}

\subsection{Logics with the classical negation}
Next we incorporate classical negation in our logics. The main result of this section shows that  the satisfiability and validity problems for  $\PLT$  are  complete for $\AEP$. The result holds also for $\PL[\mathcal{C},\clneg]$ where $\mathcal{C}$ is any finite collection of dependency atoms with  polynomial-time checkable semantics. This covers the standard dependency notions considered in the team semantics literature. The upper bound follows by an exponential-time alternating algorithm where alternation is bounded by formula depth. For the lower bound we first relate $\AEP$ to polynomial-time alternating   Turing machines that query to oracles obtained from a quantifier prefix of polynomial length. We then show how to  simulate such computations in $\PLT$.  

First we observe that the classical negation gives rise to polynomial-time reductions between the validity and the satisfiability problems. Hence, we  restrict our attention  to satisfiability hereafter.

\begin{proposition}\label{same}
Let $\varphi \in \PL[\mathcal{C},\clneg]$ where $\mathcal{C} \sub \{\dep{\cdot},\bot_{\rm c}, \subseteq\}$.  Then one can construct in polynomial time formulae $\psi,\theta \in \PL[\mathcal{C},\clneg]$ such that
\begin{enumerate}[(i)]
\item $\varphi$ is satisfiable iff $\psi$ is valid, and
\item $\varphi$ is valid iff $\theta$ is satisfiable.
\end{enumerate}
\end{proposition}
\begin{proof}
We define
\begin{align*}
\psi&:= \ma{\tuple x}  \intimp ((p\vee \neg p) \vee (\varphi\wedge {\clneg} (p\wedge \neg p))),\\
\theta& := \ma{\tuple x} \wedge ( \clneg (p\wedge \neg p) \intimp \varphi),
\end{align*}
where $\tuple x$ lists  $\var{\varphi}$. Note that $X \models  \clneg (p\wedge \neg p)$ iff $X$ is non-empty. It is straightforward to show that $(i)$ and $(ii)$ hold. Also by Proposition \ref{trans}, $\psi$ and $\theta$ can be constructed in polynomial time from $\varphi$.\qed
\end{proof}
Next we show the upper bound for the satisfiability problem of propositional logic with the classical negation, and the independence and inclusion atoms.
\begin{theorem}\label{hard}
$\SAT(\PL[\bot_{\rm c}, \sub, \clneg])\in\AEP$.
\end{theorem}
\begin{proof}
Let $\varphi \in \PL[\bot_{\rm c},\sub, \clneg]$. First existentially guess a possibly exponential-size team $T$ with domain $\var{\varphi}$. Then implement Algorithm \ref{algmc} (see Appendix) on $\Call{mc}{T,\varphi,1}$. The result follows since this algorithm  runs in polynomial time and its alternation is bounded by the size of $\varphi$. \qed

\end{proof}

Let us then turn to the lower bound. We show that the satisfiability problem of $\PLT$ is hard for $\AEP$. For  this, we first relate $\AEP$ to oracle quantification for polynomial-time oracle Turing machines. This approach is originally due to Orponen in \cite{orponen83}, where the classes $\siglevel{k}$ and $\pilevel{k}$ of the exponential-time hierarchy were characterized. 
Recall that the exponential-time hierarchy corresponds to the class of problems that can be recognized by an exponential-time  alternating Turing machine with constantly many alternations. In the next theorem we generalize Orponen's characterization to exponential-time  alternating Turing machines with \textit{polynomially} many alternations (i.e. the class $\AEP$)  by  allowing quantification of  polynomially many oracles. 

By $(A_1, \ldots ,A_k)$ we denote an efficient disjoint union of sets $A_1, \ldots ,A_k$, e.g.,  $(A_1,\ldots ,A_{k}) = \{(i,x) : x\in A_i, 1 \leq i \leq k\}$.


\begin{theorem}\label{char}
A set $A$ belongs to the class $\AEP$ iff there exist a polynomial $f$ and a polynomial-time  alternating oracle Turing machine $M$ such that, for all $ x$,
$$x \in A \textrm{ iff } Q_1 A_1 \ldots Q_{f(n)} A_{f(n)} (M\textrm{ accepts $x$ with oracles } (A_1, \ldots ,A_{f(n)}) ),$$
where $n$ is the length of $x$ and $Q_1,\dots, Q_{f(n)}$ alternate between $\exists$ and $\forall$, i.e., $Q_{i+1} \in \{\forall ,\exists\}\setminus \{Q_i\}$.
\end{theorem}
\begin{proof}
The proof is a straightforward generalization of the proof of Theorem 5.2. in \cite{orponen83}:

\textit{If-part}. Let $M$ be a polynomial-time alternating oracle Turing machine, and let $f$ and $p$ be polynomials that bound the length of the oracle quantification and the running time of $M$, respectively.
We describe the behaviour of an alternating Turing machine $M'$ such that for all $ x$,
\begin{equation}\label{M'}
M' \textrm{ accepts $x$ iff } Q_1 A_1 \ldots Q_{f(n)} A_{f(n)}  (M\textrm{ accepts $x$ with oracle } (A_1, \ldots ,A_{f(n)}) ).
\end{equation}
At first, $M'$ simulates the quantifier block $Q_1 A_1 \ldots Q_{f(n)} A_{f(n)}$ in  $f(n)$ consecutive steps.
 Namely, for $1\leq k \leq f(n)$ where $Q_k=\exists$ (or $Q_k=\forall$), $M'$ \textbf{existentially} (\textbf{universally}) chooses a set $A_k$ that consists of strings $i$ of lenght at most $ p(n)$. 
 Then $M'$ evaluates the computation tree associated with the Turing machine $M$, the input $x$, and the  selected oracle  $(A_1, \ldots ,A_{f(n)})$. In this evaluation queries to $A_k$ are replaced with investigations of the corresponding selection. We notice that $M'$ constructed in this way satisfies \eqref{M'}, alternates $f(n)$ many times, and runs in time $2^{h(n)}$ for some polynomial $h$.

\textit{Only-if part}. Let  $M'$ be an alternating exponential-time Turing machine  with  polynomially many alternations. We show how to construct an alternating polynomial-time oracle Turing machine $M$ satisfying \eqref{M'}. 
W.l.o.g. we find polynomials $f$ and $g$ such that  $M'$ runs in time  at least $n$ and at most $2^{f(n)}-2$ and has at most $g(n)$ many alternations. 

Let $\#$ be a symbol that is not in the alphabet and denote $2^{f(n)}-1$ by $m$. Each configuration of $M'$ can be represented as a string
$$  \alpha =   u q  v\# \ldots \#, |\al  | = m,$$
with the meaning that $M'$ is in state $q$, has string $ u v$ on its tape, and reads the first symbol of string $ v$. The symbol $\#$ is only used to pad configurations to the same length. A computation of $M'$ over $ x$ may be represented as a sequence of configurations $\al_0, \al_1, \ldots , \al_{m}$ such that $ \al_0 = q_0x \#\ldots \#$ where $q_0$ is the initial state, $ \al_{m}=  u q  v\#\ldots \#$ where  $q$ is some final state, and for  $i\leq m-1$ either $ \al_{i+1}$ is reachable from $ \al_i$ with one step or $\al_i= \al_{i+1}= \al_m$.
Each oracle $A_k$ can encode a  computation sequence $\al ^k_0, \al^k_1, \ldots , \al^k_{m}$ with triples $(i,j, \al^k_{i,j})$ where  $|i|,|j|\leq f(n)$ and $ \al^k_{i,j}$ is the  $j$th symbol  of configuration  $ \al^k_{i}$.  Determining whether $k,i,j$ generate a unique $\al^k_{i,j}$ can be done with a bounded number of $A_k$ queries since there are only finitely many alphabet and state symbols in $M'$.

Next we describe the behaviour of the alternating polynomial-time oracle Turing machine $M$. The idea is to simulate the computation of $M'$ using the above succinct encoding. 
$M$ proceeds in $g(n)$ consecutive steps, and below we present step $k$ for  $1 \leq k \leq g(n)$ and $Q^k=\exists$. Notice that we use $v$ to indicate the last  alternation point  of $M'$, i.e., $v$ is a binary string that is initially set to $0$ and has always  length at most $ f(n)$. 
 Notice also that by $\al^0_{0,j}$ we refer to the $j$th symbol of configuration  $ \al_0 = q_0x \#\ldots \#$.\\

\vspace{-2mm}
\textbf{step $k$}:
\begin{compactenum}
\vspace{1mm}
\item \textbf{universally guess} $i,j$ such that $|i|,|j|\leq f(n)$ and $v \leq  i $;
\begin{enumerate}
\item[(1a)] \textbf{if} $ \al^{k-1}_{v,j}=\al^{k}_{v,j} $ and  $\al^k_{i,j-1},\al^k_{i,j},\al^k_{i,j+1},\al^k_{i,j+2}$ correctly determine $\al^k_{i+1,j}$ \textbf{then} proceed to (2);
\item[(1b)] \textbf{otherwise} return \texttt{false};
\end{enumerate}
\item \textbf{existentially guess} $w$ such that $|w|\leq f(n)$ and $v < w$;
\item  \textbf{universally guess}  $i,j$ such that $|i|,|j|\leq f(n)$ and $v< i < w$;
\begin{enumerate}
\item[(3a)] \textbf{if} $\al^k_{i,j}$ is not a universal state \textbf{then} proceed to (4);
\item[(3b)] \textbf{otherwise} return \texttt{false};
 \end{enumerate} 
 \item  \textbf{existentially guess} $j$ such that $|j|\leq f(n)$;
 \begin{enumerate} 
 \item[(4a)] \textbf{if} $w < m$ and $\al^k_{w,j}$ is a universal state \textbf{then} set $v\gets w$ and proceed to \textbf{step $k+1$};
 \item[(4b)] \textbf{else if} $w=m$ and $\al^k_{w,j}$ is an accepting state \textbf{then} return \texttt{true};
 \item[(4c)] \textbf{otherwise} return \texttt{false}.
 \end{enumerate}
 \vspace{-1mm}
\end{compactenum}
~\\
For $1 \leq k \leq g(n)$ and $Q^k=\exists$, step $k$ is described as the dual of the above procedure. Namely, it is obtained by replacing in item (1) universal guessing with existential one,  in item (1b) false with true, and in items (3a) and (4a) universal state with existential state. It is now straightforward to check that $M$ runs in polynomial time  and satisfies \eqref{M'}. 
\qed
\end{proof}

Using this theorem we now prove Theorem \ref{main}. For the quantification over oracles $A_i$, we use repetitively $\vee$ and $\clneg$. 

\begin{theorem}\label{main}
$\SAT(\PL[\clneg])$ is hard for $\AEP$.
\end{theorem}
\begin{proof}
Let $A \in \AEP$. From Theorem \ref{char} we obtain 
a polynomial $f$ and  an alternating oracle Turing machine $M$  with running time bounded by $g$. 
By \cite{chandra81}, the alternating machine can be replaced by a sequence of word quantifiers over a deterministic Turing machine. (Strictly speaking, \cite{chandra81} speaks only about a bounded number of alternations, but the generalization to the unbounded case is straightforward.)
W.l.o.g. we may assume  that 
each configuration of $M$ has at most two configurations reachable in one step. It then follows by Theorem \ref{char} that one can construct
a polynomial-time  deterministic oracle Turing machine $M^*$ such that $x\in A$ iff
\begin{align*}
Q_1 A_1 \ldots Q_{f(n)} A_{f(n)} &Q'_{1} \tuple   y_1\ldots Q'_{g(n)} \tuple y_{g(n)} 
\\&\textrm{($M^*$ accepts $(x,\tuple y_1, \ldots ,\tuple y_{g(n)})$ with oracle }
(A_1, \ldots ,A_{f(n)})),
\end{align*} 
where   $Q_1,\dots, Q_{f(n)}$ and $Q'_{1} ,\ldots ,Q'_{g(n)}$  are alternating sequences of quantifiers $\exists$ and $\forall$, and each $\tuple y_i$ is a $g(n)$-ary sequence of propositional symbols where $n$ is the length of $x$. Note that  $M^*$ runs in polynomial time also with respect to $n$. Using this characterization we now show how to reduce in polynomial time any $x$  to a formula $\varphi$ in $\PL[\clneg]$ such that $x\in A$ iff $\varphi$ is satisfiable. We construct $\varphi$ inductively. As a first step, we let 
$$\varphi := \ma{\tuple q \tuple r \tuple y} \wedge p_t \wedge \neg p_f \wedge \varphi_1 $$
where 
\begin{itemize}
\item $\tuple q$ and $\tuple r$ list propositional symbols that are used for encoding oracles;
\item $\tuple y$ lists propositional symbols that occur in $\tuple y_1, \ldots ,\tuple y_{g(n)}$ and in  $\tuple z_i$ that are used to simulate configurations of $M^*$ (see phase (3) below);
\item $p_t$ and $p_f$ are propositional symbols that do not occur in $\tuple q\tuple r \tuple y$.
\end{itemize}

\paragraph{\textbf{(1) Quantification over oracles}}
Next we show how to simulate quantification  over oracles.  W.l.o.g. we may assume that $M^*$ queries binary strings that are of length $h(n)$ for some polynomial $h$. Let $\tuple q$ be a sequence of length $h(n)$ and $\tuple r$ a sequence of length $f(n)$. Our intention is that $\tuple q$ with $r_i$ encodes the content of the oracle $A_i$; in fact $\tuple q$ and $r_i$ encode the characteristic function of the relation that corresponds to the oracle $A_i$.   
For a string of bits $\tuple b=b_1\ldots b_k$ and a sequence $\tuple s=(s_1, \ldots ,s_k)$ of proposition symbols, we write $\tuple s = \tuple b$ for $\bigwedge_{i=1}^k s_{i}^ {b_i}$, where $s_i^1  := s_i$ and $s_i^0 := \neg s_i$. The idea is that, given a team $X$ over $\tuple q\,\tuple r$, an oracle $A_i$, and a binary string $\tuple a =a_1\ldots a_{h(n)}$, the membership of $\tuple a$ in $ A_i$ is expressed by $X \models  \clneg \neg (\tuple q = \tuple a \wedge r_i)$.  Note that the latter indicates that there exists  $s \in X$ mapping $\tuple q \mapsto \tuple a$ and $r_i \mapsto 1$. Following this idea we next show how to simulate quantification over oracles $A_i$. We define $\varphi_{i}$, for $1 \leq i \leq f(n)$, inductively from root to leaves.  Depending on whether $A_i$ is existentially or universally quantified, we let
\begin{itemize}
\item[$\exists$:] $\varphi_{i}:= \dep{\vec{q}, r_i} \vee (\dep{\vec{q},r_i} \wedge \varphi_{i+1})$,
\item[$\forall$:] $\varphi_{i}:= \hspace{.5mm}\clneg \dep{\vec{q}, r_i} \cwedge (\clneg \dep{\vec{q}, r_i} \cvee \varphi_{i+1})$.
\end{itemize}
The formula $\varphi_{f(n)+1}$ will be  $\psi_1$  defined in step (2) below. 
Let us explain the idea behind the definitions of $\varphi_i$, first in the case of existential quantification. Assume that $X$ is a team such that 
\begin{equation}\label{eq1}
X \models \dep{\vec{q}, r_i} \vee (\dep{\vec{q},r_i} \wedge \varphi_{i+1}),
\end{equation}
and, for $j\geq i$,  $X$ is maximal over $r_{j}$ for all $\vec{z}_{j}$, where $\vec{z}_{j}$ lists all symbols from the domain of $X$ except $r_{j}$. Then by \eqref{eq1} we may  choose  two subsets $Y,Z \sub X$,  $Y\cup Z=X$, where $Y \models \dep{\vec{q}, r_i}$ and $Z \models \dep{\vec{q}, r_i} \wedge \varphi_{i+1}$. Note that since especially $X$ was maximal over $r_{i}$ for all $\vec{q}$, the selection of the partition $Y\cup Z=X$ essentially quantifies over the characteristic functions of the oracle $A_i$. Moreover, note that,  for $j\geq i+1$,  $Z$ is maximal over $r_{j}$ for all $\vec{z}_{j}$, where $\vec{z}_{j}$ is defined as above.

Universal quantification is simulated analogously. This time we have that
\begin{equation}\label{eq1.1}
X \models \clneg \dep{\vec{q}, r_i} \cwedge (\clneg \dep{\vec{q}, r_i} \cvee \varphi_{i+1}),
\end{equation}
and range over all subsets $Y,Z \sub X$ where  $Y\cup Z=X$. By \eqref{eq1.1}  for all such $Y$ and $Z$, we have that if $Y \models \dep{\vec{q},r_i}$ and $Z \models \dep{\vec{q},r_i}$ then $Z \models \varphi_{i+1}$ (see Section \ref{ao} for the definition of  $ \cwedge$). Using an analogous argument for $Z$ as in the existential case, we notice that the selection of $Z$ corresponds to  universal quantification over  characteristic functions of  $A_i$.

\paragraph{\textbf{(2) Quantification over propositional symbols}} 
Next we show how to simulate the quantifier block $Q'_1 \tuple y_1  \ldots Q'_{g(n)}\tuple y_{g(n)} \exists \tuple z$ where $\tuple z$ lists all propositional symbols that occur in $\tuple y$ but not in any $\tuple y_i$ (i.e. the remaining symbols that occur when simulating $M^*$). Assume that this quantifier block is of the form $Q^*_1 y_1 \ldots Q^*_l y_l$, and let $\psi_1:= \varphi_{f(n)+1}$. We define $\psi_i$ again top-down inductively. For $1 \leq i \leq l$, depending on whether $Q^*_i$ is $\exists$ or $\forall$, we let
\begin{itemize}
\item[$\exists$:] $\psi_{i}:= \dep{y_i} \vee (\dep{y_i} \wedge \psi_{i+1})$,
\item[$\forall$:] $\psi_{i}:= \hspace{.5mm}\clneg \dep{y_i} \cwedge (\clneg \dep{y_i}\cvee \psi_{i+1})$.
\end{itemize}
Let us explain the idea behind the two definitions of $\psi_i$. The idea is essentially the same as in the oracle quantification step. First in the case of existential quantification. Assume that we consider a formula $\psi_i$ and a team $X$ where \begin{equation}\label{eq2}
X \models \psi_i,
\end{equation}
and  $X$ is maximal over $ y_i \ldots y_l$ for all $\tuple q \tuple r y_1\ldots  y_{i-1}$. By \eqref{eq2} we may choose two subsets $Y,Z \sub X$,  $Y\cup Z=X$, where $Y\models  \dep{y_i}$ and $Z \models  \dep{y_i} \wedge \psi_{i+1}$. There are now two options: either we choose $Z= \{s\in X\mid s(y_i) =0\}$ or $Z= \{s\in X\mid s(y_i) =1\}$. Since $X$ is maximal over $ y_i \ldots y_l$ for all $\tuple q \tuple r y_1\ldots  y_{i-1}$, we obtain that $Z \upharpoonright \tuple q \tuple r = X \upharpoonright \tuple q \tuple r$ and $Z$ is maximal over $y_{i+1}\ldots y_l$ for all $\tuple q \tuple r y_1 \ldots y_i$. Hence no information about  oracles is lost in this quantifier step.

The case of universal quantification is again analogous to the oracle case. Hence we obtain that $\eqref{eq2}$ holds iff both $\{s\in X\mid s(y_i) =0\}$ and $\{s\in X\mid s(y_i) =1\}$ satisfy $\psi_{i+1}$.

\paragraph{\textbf{(3) Simulation of computations}}
Next we define $\psi_{g(n)+1}$ that simulates the polynomial-time deterministic oracle Turing machine $M^*$. Note that this formula is evaluated over a subteam $X$ such that $X \models \dep{y_i}$, for each $y_i\in \tuple y$, and $\tuple a \in A_i$ iff $X \models  \clneg \neg (\tuple q = \tuple a \wedge r_i)$. Using this it is now straightforward to construct a propositional formula $\theta$ such that 
$\exists  \tuple c (X[\tuple b_i/\tuple y_i][\tuple c / \tuple z] \models \theta)$ if and only if  $M^*$ accepts $(x,\tuple b_1, \ldots ,\tuple b_{g(n)})$ with oracle $(A_1, \ldots ,A_{f(n)})$.
Here $X[\tuple a/\tuple x]$ denotes the team $\{s(\tuple a/\tuple x):s\in X\}$ where $s(\tuple a / \tuple x)$ agrees with $s$ everywhere except that it maps $x$ pointwise to $\tuple a$. Each configuration of $M^*$ can be encoded with a binary sequence $\tuple z_i$ of length $O(t(n))$ where $t$ is a polynomial bounding the running time of $M^*$. Then it suffices to define $\psi_{l+1}$ as a conjunction of formulae $\theta_{\rm start}(\tuple z_0), \theta_{\rm move}(\tuple z_i,\tuple z_{i+1}), \theta_{\rm final}(\tuple z_{t(n)})$ describing that $\tuple z_0$ corresponds to the initial configuration, $\tuple z_i$ determines $\tuple z_{i+1}$, and $\tuple z_{t(n)}$ is in accepting state. Note that the formulae $\theta_{\rm start}(\tuple z_0)$, $\theta_{\rm move}(\tuple z_i,\tuple z_{i+1})$, and $\theta_{\rm final}(\tuple z_{t(n)})$ can be written exactly as in the classical setting, except that all disjunctions $\lor$ are replaced by the intuitionistic disjunction $\varovee$.

Finally note that, by Proposition \ref{trans}, all occurrences of dependence atoms, the shorthand $\max(\cdot)$, and the connectives $\varovee$ and $\cwedge$ can be eliminated from the above formulae by a polynomial overhead. Thus the constructed formula $\varphi$ is a $\PL[\clneg]$-formula as required.
\end{proof}

By Proposition \ref{same}, and   Theorems \ref{hard} and  \ref{main} we now obtain the following.
\begin{theorem}
Satisfiability and validity of $\PL[\bot_{\rm c}, \sub, \clneg]$ and $\PL[\clneg]$ are complete for $\AEP$.
\end{theorem}
The following corollary now follows by a direct generalisation of Theorem \ref{hard}.
\begin{corollary}
Let $\mathcal{C}$ be a finite collection of dependency atoms with polynomial-time checkable semantics. Satisfiability and validity of $\PL[\mathcal{C}, \clneg]$ is complete for $\AEP$.
\end{corollary}

\section{Complexity of Model Checking}

In this section we consider the  model checking problems of our logics. We first focus on logics without the  classical negation. 
\begin{theorem}\label{mcind}
$\MC(\PLI)$ is complete for $\NP$.
\end{theorem}
\begin{proof}
The upper bound follows since the model checking problem for modal independence logic is $\NP$-complete \cite{DBLP:journals/corr/KontinenMSV14a}. Since dependence atoms can be expressed efficiently by independence atoms (see the proof of Theorem \ref{valPLIhard}), the lower bound follows from the $\NP$-completeness of $\MC(\PD)$ (see Table \ref{table:PLoldresults}).
\end{proof}
The following unpublished result was shown by Hella et al.
\begin{theorem}[\cite{hkmv16}]
$\MC(\PLInc)$ is $\PTIME$-complete.
\end{theorem}
The following result can also be found in the PhD thesis of M\"uller \cite{muller14}.
\begin{theorem}
$\MC(\PL[\clneg])$ is complete for $\PSPACE$.
\end{theorem}
\begin{proof}
For the upper bound note that Algorithm \ref{algmc} decides the problem in $\AEPoly$ which is exactly $\PSPACE$ \cite{chandra81}. For the lower bound, we reduce from $\TQBF$ which is known to be $\PSPACE$-complete. Let $Q_1 x_1 \ldots Q_n x_n \theta$ be a quantified boolean formula. Let $\tuple r$ be a sequence of propositional symbols of length $\log (n)+1$, and let $T:= \{s_1, \ldots ,s_n\}$ be a team  where $s_i(\tuple r)$ writes $i$ in binary. 
   We define inductively  a formula $\varphi \in \PL[\clneg]$ such that 
\begin{equation}\label{plaa}
Q_1 x_1 \ldots Q_n x_n \theta \textrm{ is true iff }T \models \varphi.
\end{equation}
Let $\varphi:= \varphi_1$, and for $1 \leq i \leq n$, depending on whether $x_i$ is existentially or universally quantified we let
\begin{itemize}
\item[$\exists$:] $\varphi_{i}:= \tuple r = \cor{i} \vee  \varphi_{i+1}$,
\item[$\forall$:] $\varphi_{i}:= \hspace{.5mm}\clneg \tuple r = \cor{i} \cwedge  \varphi_{i+1}$.
\end{itemize}
Finally, we let $\varphi_{n+1}$ denote the formula obtained from $\theta$ by first substituting each $\neg x_i$ by  $\neg \tuple r = \cor{i}$ and then  $x_i$ by $\clneg \neg \tuple r = \cor{i}$, for each $i$. Note that the meaning $\neg \tuple r = \cor{i}$ is that the assignment $s_i$ is not in the team, whereas $\clneg \neg \tuple r = \cor{i}$ states that $s_i$ is in the team. It is now straightforward to establish that \eqref{plaa} holds. Also $T$ and $\varphi$ can be constructed in polynomial time, and hence we obtain the result. \qed
\end{proof}
Since Algorithm \ref{algmc} can also be applied to independence and inclusion atoms, we obtain the following corollary.
\begin{corollary}
$\MC(\PL[\bot_{\rm c}, \sub, \clneg])$ and $\MC(\PL[\mathcal{C}, \clneg])$, where $\mathcal{C}$ is a finite collection of polynomial time computable dependency atoms, are complete for $\PSPACE$.
\end{corollary}

\section{Conclusion}
In this article we have initiated a systematic study of the complexity theoretic properties of  team based propositional logics. Regarding the logics considered in this paper, an interesting open question is to determine the exact complexity of  $\VAL(\PLI)$ for which membership in 
$\coNEXP^{\NP}$ was shown in this paper. Propositional team semantics is a very rich framework in which many interesting connectives and operators  can be studied such as the intuitionistic implication $\intimp$ applied in the area of inquisitive semantics. It is an interesting question to extend this study to cover a more wide range of team based logics.



\bibliographystyle{plain}
\bibliography{biblio}

\begin{thebibliography}{10}

\bibitem{abramsky09}
Samson Abramsky and Jouko V\"a\"an\"anen.
\newblock From {IF} to {BI}.
\newblock {\em Synthese}, 167:207--230, 2009.
\newblock 10.1007/s11229-008-9415-6.

\bibitem{Buss1987}
Sam Buss.
\newblock The boolean formula value problem is in {ALOGTIME}.
\newblock In {\em Proceedings of the Nineteenth Annual ACM Symposium on Theory
  of Computing}, STOC '87, pages 123--131, New York, NY, USA, 1987. ACM.

\bibitem{chandra81}
Ashok~K. Chandra, Dexter~C. Kozen, and Larry~J. Stockmeyer.
\newblock Alternation.
\newblock {\em J. ACM}, 28(1):114--133, January 1981.

\bibitem{co71}
Stephen~A. Cook.
\newblock The complexity of theorem-proving procedures.
\newblock In {\em Proceedings of the third annual ACM symposium on Theory of
  computing}, STOC '71, pages 151--158, New York, NY, USA, 1971. ACM.

\bibitem{Durand2016}
Arnaud Durand, Juha Kontinen, and Heribert Vollmer.
\newblock Expressivity and complexity of dependence logic.
\newblock In Samson Abramsky, Juha Kontinen, Jouko V{\"a}{\"a}n{\"a}nen, and
  Heribert Vollmer, editors, {\em Dependence Logic: Theory and Applications},
  pages 5--32. Springer International Publishing, 2016.

\bibitem{ebbing12}
Johannes Ebbing and Peter Lohmann.
\newblock Complexity of model checking for modal dependence logic.
\newblock In M{\'a}ria Bielikov{\'a}, Gerhard Friedrich, Georg Gottlob, Stefan
  Katzenbeisser, and Gy{\"o}rgy Tur{\'a}n, editors, {\em SOFSEM 2012: Theory
  and Practice of Computer Science}, volume 7147 of {\em Lecture Notes in
  Computer Science}, pages 226--237. Springer Berlin / Heidelberg, 2012.

\bibitem{galliani12}
Pietro Galliani.
\newblock Inclusion and exclusion dependencies in team semantics: On some
  logics of imperfect information.
\newblock {\em Annals of Pure and Applied Logic}, 163(1):68 -- 84, 2012.

\bibitem{DBLP:journals/sLogica/GradelV13}
Erich Gr{\"a}del and Jouko V{\"a}{\"a}n{\"a}nen.
\newblock Dependence and independence.
\newblock {\em Studia Logica}, 101(2):399--410, 2013.

\bibitem{DBLP:conf/foiks/HannulaK14}
Miika Hannula and Juha Kontinen.
\newblock A finite axiomatization of conditional independence and inclusion
  dependencies.
\newblock {\em Inf. Comput.}, 249:121--137, 2016.

\bibitem{hakoluvi16}
Miika Hannula, Juha Kontinen, Martin L\"uck, and Jonni Virtema.
\newblock On quantified propositional logics and the exponential time
  hierarchy.
\newblock In Domenico Cantone and Giorgio Delzanno, editors, {\em {\rm
  Proceedings of the Seventh International Symposium on} Games, Automata,
  Logics and Formal Verification, {\rm Catania, Italy, 14-16 September 2016}},
  volume 226 of {\em Electronic Proceedings in Theoretical Computer Science},
  pages 198--212. Open Publishing Association, 2016.

\bibitem{hkmv16}
Lauri Hella, Antti Kuusisto, Arne Meier, and Jonni Virtema.
\newblock Model checking and validity in propositional and modal inclusion
  logics.
\newblock {\em CoRR}, abs/1609.06951, 2016.

\bibitem{hellakmv15}
Lauri Hella, Antti Kuusisto, Arne Meier, and Heribert Vollmer.
\newblock Modal inclusion logic: Being lax is simpler than being strict.
\newblock In Giuseppe~F. Italiano, Giovanni Pighizzini, and Donald Sannella,
  editors, {\em MFCS (1)}, volume 9234 of {\em Lecture Notes in Computer
  Science}, pages 281--292. Springer, 2015.

\bibitem{DBLP:journals/corr/KontinenMSV14a}
Juha Kontinen, Julian-Steffen M{\"u}ller, Henning Schnoor, and Heribert
  Vollmer.
\newblock {A Van Benthem Theorem for Modal Team Semantics}.
\newblock In {\em 24th EACSL Annual Conference on Computer Science Logic (CSL
  2015)}, pages 277--291, 2015.

\bibitem{DBLP:journals/fuin/KontinenN11}
Juha Kontinen and Ville Nurmi.
\newblock Team logic and second-order logic.
\newblock {\em Fundam. Inform.}, 106(2-4):259--272, 2011.

\bibitem{Lev1973}
Leonid~A. Levin.
\newblock Universal search problems.
\newblock {\em Problems of Information Transmission}, 9(3), 1973.

\bibitem{lohvo13}
Peter Lohmann and Heribert Vollmer.
\newblock Complexity results for modal dependence logic.
\newblock {\em Studia Logica}, 101(2):343--366, 2013.

\bibitem{Luck16a}
Martin L{\"u}ck.
\newblock {Axiomatizations for Propositional and Modal Team Logic}.
\newblock In Jean-Marc Talbot and Laurent Regnier, editors, {\em 25th EACSL
  Annual Conference on Computer Science Logic (CSL 2016)}, volume~62 of {\em
  Leibniz International Proceedings in Informatics (LIPIcs)}, pages
  33:1--33:18, Dagstuhl, Germany, 2016. Schloss Dagstuhl--Leibniz-Zentrum fuer
  Informatik.

\bibitem{Luck16b}
Martin L{\"{u}}ck.
\newblock Complete problems of propositional logic for the exponential
  hierarchy.
\newblock {\em CoRR}, abs/1602.03050, 2016.

\bibitem{muller14}
Julian-Steffen M\"uller.
\newblock Satisfiability and model checking in team based logics.
\newblock {\em PhD Thesis, Leibniz Universit\"at Hannover, Cuvillier Verlag
  G\"ottingen}, 2014.

\bibitem{orponen83}
Pekka Orponen.
\newblock Complexity classes of alternating machines with oracles.
\newblock In {\em Automata, Languages and Programming, 10th Colloquium,
  Barcelona, Spain, July 18-22, 1983, Proceedings}, pages 573--584, 1983.

\bibitem{DBLP:journals/corr/SanoV14}
Katsuhiko Sano and Jonni Virtema.
\newblock {Axiomatizing Propositional Dependence Logics}.
\newblock In Stephan Kreutzer, editor, {\em 24th EACSL Annual Conference on
  Computer Science Logic (CSL 2015)}, volume~41 of {\em Leibniz International
  Proceedings in Informatics (LIPIcs)}, pages 292--307, Dagstuhl, Germany,
  2015. Schloss Dagstuhl--Leibniz-Zentrum fuer Informatik.

\bibitem{vaananen07}
Jouko V\"a\"an\"anen.
\newblock {\em Dependence Logic}.
\newblock Cambridge University Press, 2007.

\bibitem{DBLP:journals/corr/Virtema14}
Jonni Virtema.
\newblock Complexity of validity for propositional dependence logics.
\newblock In Adriano Peron and Carla Piazza, editors, {\em Proceedings Fifth
  International Symposium on Games, Automata, Logics and Formal Verification,
  GandALF 2014, Verona, Italy, September 10-12, 2014.}, volume 161 of {\em
  {EPTCS}}, pages 18--31, 2014.

\bibitem{Yangthesis}
Fan Yang.
\newblock {\em On Extensions and Variants of Dependence Logic}.
\newblock PhD thesis, University of Helsinki, 2014.

\bibitem{DBLP:journals/corr/YangV14}
Fan Yang and Jouko V{\"{a}}{\"{a}}n{\"{a}}nen.
\newblock Propositional logics of dependence.
\newblock {\em Ann. Pure Appl. Logic}, 167(7):557--589, 2016.

\end{thebibliography}

\newpage
\pagestyle{empty}
\appendix
\setcounter{section}{1}
\begin{algorithm}
			\caption{$\AEPoly$ algorithm for $\MC(\PL[\bot_{\rm c}, \sub,\clneg])$}
			\label{algmc}
		\begin{algorithmic}[1]
			\Function{mc}{$T, \varphi, I$} 
			\If{$\varphi = \psi_1 \wedge \psi_2$}
				\If{$I=1$}
					\State \textbf{universally choose} $i \in \{1,2\}$
					\State \textbf{return} \Call{mc}{$T,\psi_i,I$}
				\ElsIf{$I=0$}
					\State \textbf{existentially choose} $i \in \{1,2\}$
					\State \textbf{return} \Call{mc}{$T,\psi_i,I$}
				\EndIf

			\ElsIf{$\varphi = \psi_1 \vee \psi_2$}
				\If{$I=1$}
					\State \textbf{existentially choose} $T_1\cup T_2=T$
					\State \textbf{universally choose} $i\in \{1,2\}$
					\State \textbf{return} \Call{mc}{$T_i,\psi_i,I$}
				\ElsIf{$I=0$}
					\State \textbf{universally choose} $T_1\cup T_2=T$
					\State \textbf{existentially choose} $i\in \{1,2\}$
					\State \textbf{return} \Call{mc}{$T_i,\psi_i,I$} 
				\EndIf

			\ElsIf{$\varphi = \hspace{.5mm}\clneg \psi$}
				\If{$I=1$}
					\State \textbf{return} \Call{mc}{$T,\psi,0$}
				\ElsIf{$I=0$}
					\State \textbf{return} \Call{mc}{$T,\psi,1$}
				\EndIf

			\ElsIf{$\varphi = p$ ($\varphi = \neg p$)}
				\State $1 \leftarrow x$
				\For{$s \in T$}
					\If{$s(p)=0$ ($s(p)=1$)}
						\State  $0 \leftarrow x$	
					\EndIf
				\EndFor	
				\If{$x=I=1$ \textbf{or} $x=I=0$}
					\State \textbf{return} \textbf{true}
				\Else
					\State \textbf{return} \textbf{false}
				\EndIf


			\ElsIf{$\varphi = \tuple p \sub \tuple q$}
				\State $1 \leftarrow x$
				\For{$s \in T$}
						\State $0\leftarrow y$
						\For{$s' \in T$}
							\If{$s(\tuple p) = s'(\tuple q)$}
								\State $1\leftarrow y$
							\EndIf
						\EndFor
						\If{$y=0$}
							\State $0 \leftarrow x$
						\EndIf 
				\EndFor
				\If{$x=I=1$ \textbf{or} $x=I=0$}
				\State \textbf{return} \textbf{true}
				\Else
				\State \textbf{return} \textbf{false}
				\EndIf

			\ElsIf{$\varphi = \indep{\tuple p}{\tuple q}{\tuple r}$}
				\State $1 \leftarrow x$
				\For{$s,s' \in T \textbf{ with } s(\tuple p)=s'(\tuple p)$}
					\State $0\leftarrow y$
					\For{$s'' \in T$}
							\If{$s(\tuple p) = s''(\tuple p),s(\tuple q) = s''(\tuple q), s'(\tuple r)=s''(\tuple r)$}
								\State $1\leftarrow y$
							\EndIf
						\EndFor
						\If{$y=0$}
							\State $0 \leftarrow x$
						\EndIf 
				\EndFor
				\If{$x=I=1$ \textbf{or} $x=I=0$}
				\State \textbf{return} \textbf{true}
				\Else
				\State \textbf{return} \textbf{false}
				\EndIf

%
%
%

			\EndIf
			\EndFunction
		\end{algorithmic}
	\end{algorithm}
\end{document}